\documentclass[letterpaper, 10 pt, journal, twoside]{ieeetran}

\usepackage[pdftex, pdfstartview={FitV}, pdfpagelayout={TwoColumnLeft},bookmarksopen=true,plainpages = false, colorlinks=true, linkcolor=black, citecolor = black, urlcolor = black,filecolor=black , pagebackref=false,hypertexnames=false, plainpages=false, pdfpagelabels ]{hyperref}

\usepackage{setspace}
\usepackage{graphicx}
\usepackage{epstopdf}
\usepackage{amssymb,amsmath}

\usepackage{amsthm}
\usepackage{mathtools}
\usepackage{cuted}
\usepackage{cases}
\usepackage[font=footnotesize]{subcaption}

\usepackage[font=footnotesize]{caption}
\usepackage{xcolor}
\usepackage{units}
\usepackage{algorithmic}
\usepackage[linesnumbered, ruled, vlined]{algorithm2e}
\usepackage{balance}
\usepackage[sort,compress,noadjust]{cite}
\usepackage{tabularx}
\usepackage{nameref}
\usepackage{centernot}

\usepackage{arydshln}
\newtheoremstyle{definition}{}{}{}{}{\bfseries}{.}{.5em}{\thmname{#1}\thmnumber{ #2}\thmnote{ (#3)}}
\theoremstyle{definition}
\newtheorem{definition}{Definition}

\usepackage{enumitem}

\theoremstyle{plain}
\newtheorem{theorem}{Theorem}
\theoremstyle{plain}
\newtheorem{proposition}{Proposition}
\theoremstyle{plain}
\newtheorem{lemma}{Lemma}
\theoremstyle{plain}
\newtheorem{corollary}{Corollary}
\theoremstyle{definition}
\newtheorem{assumption}{Assumption}
\theoremstyle{remark}
\newtheorem{remark}{Remark}

\usepackage[hang,flushmargin]{footmisc}

\usepackage[capitalize]{cleveref}
\crefformat{equation}{(#2#1#3)}
\Crefformat{equation}{Equation~(#2#1#3)}
\Crefname{equation}{Equation}{Eqs.}

\usepackage{accents}

\usepackage[authormarkuptext=name,addedmarkup=bf,authormarkupposition=right]{changes}
\definechangesauthor[name={BTL}, color={blue}]{bl}

\title{
Dynamic Adaptation Gains for Nonlinear Systems with Unmatched Uncertainties
}

\author{Brett T. Lopez$^{1}$ and Jean-Jacques Slotine$^{2}$
\thanks{$^{1}$Verifiable and Control-Theoretic Robotics Laboratory, University of California, Los Angeles, Los Angeles CA, {\tt\small btlopez@ucla.edu}}
\thanks{$^{2}$Nonlinear Systems Laboratory, Massachusetts Institute of Technology, Cambridge MA, {\tt\small jjs@mit.edu}}
}

\begin{document}

\maketitle
\thispagestyle{empty}
\pagestyle{empty}

\begin{abstract}
We present a new direct adaptive control approach for nonlinear systems with unmatched and matched uncertainties. 
The method relies on adjusting the adaptation gains of individual unmatched parameters whose adaptation transients would otherwise destabilize the closed-loop system.
The approach also guarantees the restoration of the adaptation gains to their nominal values and can readily incorporate direct adaptation laws for matched uncertainties.
The proposed framework is general as it only requires stabilizability for all possible models.
\end{abstract}

\section{Introduction}

\IEEEPARstart{A}{daptive} control of systems with unmatched uncertainties, i.e., model perturbations outside the span of the control input matrix, is notoriously difficult because the uncertainties cannot be directly canceled by the control input.
This is especially true for nonlinear systems where combining a stable model estimator with a nominal feedback controller does not necessarily yield a stable closed-loop system without imposing limits on the growth rate of the uncertainties, so as to prevent finite escape. 
The prevailing approach for handling nonlinear systems with unmatched uncertainties has been to construct nominal feedback controllers that are robust to time-varying parameter estimates \cite{krstic1995nonlinear}.
More precisely, one seeks to construct an input-to-state stable control Lyapunov function (ISS-clf) \cite{sontag1995characterizations} which, by construction, ensures the system converges to a region near the desired state despite parameter estimation error and transients.
Once an ISS-clf and corresponding controller are known, then any stable model estimator can be employed without concern of instability.  
From a theoretical standpoint, this framework guarantees stable closed-loop control and estimation despite unmatched uncertainties but requires constructing an ISS-clf -- a nontrivial task unless the system takes a particular structure.

The approach discussed above is categorized as \textit{indirect} adaptive control and entails combining a robust controller with a stable model estimator. 
Conversely, \textit{direct} adaptive control uses Lyapunov-like stability arguments to construct a parameter adaptation law that guarantees the state converges to a desired value.
This eliminates the robustness requirements inherent to indirect adaptive control which, in some ways, simplifies design.
However, since unmatched uncertainties cannot be directly canceled through control, it is necessary to construct a \emph{family} of Lyapunov functions that depend on the estimates of the unmatched parameters.
This model dependency introduces sign-indefinite terms (related to the parameter estimation transients) in the stability proof that require sophisticated direct adaptive control schemes to achieve stable closed-loop control and adaptation. 
Adaptive control Lyapunov functions (aclf) where proposed to cancel the problematic transient terms by synthesizing a family of clf's for a modified dynamical system that depends on the family of clf's \cite{krstic1995control}.
Recently, \cite{lopez2021universal} proposed a method that adjusts the adaptation gain online to cancel the undesirable transient terms and requires no modifications to the standard clf definition or dynamical system of interest. 

The main contribution of this work is a new direct adaptive control methodology which generalizes and improves the online adaptation gain adjustment approach from \cite{lopez2021universal}, subsequently called \emph{dynamic adaptation gains}, yielding a superior framework for handling various forms of model uncertainties.
{The new framework has two distinguishing properties: 1) adaptation gains are \emph{individually} adjusted online to prevent parameter adaptation transients from destabilizing the system and 2) the adaptation gain is \emph{guaranteed} to return to its nominal value asymptotically.}
Conceptually, the adaptation gain is lowered, i.e., adaptation slowed, for parameters whose transients is destabilizing; the adaptation gain then returns to its nominal once the transients subsides. 
Noteworthy properties are examined, e.g., handling of matched uncertainties, and several forms of the dynamic adaptation gain update law are derived.
The approach only relies on stabilizability of the uncertain system so it is a very general framework.
Simulation results of a nonlinear system with unmatched and matched uncertainties demonstrates the approach.

\textit{Notation:} The set of positive and strictly-positive scalars will be denoted as $\mathbb{R}_{\geq 0}$ and $\mathbb{R}_{>0}$, respectively.
The shorthand notation for a function $T$ parameterized by a vector $a$ with vector argument $s$ will be $T_a(s) = T(s;a)$.
The partial derivative of a function $N(x,y)$ will sometimes denoted as $\nabla_x N(x,y) = \frac{\partial N}{\partial x}$ where the subscript on $\nabla$ is omitted when there is no ambiguity.

\section{Problem Formulation}
This work addresses control of uncertain dynamical systems of the form
\begin{equation}
    \dot{x} = f(x,t) - \Delta(x,t)^\top \theta + B(x,t) u,
    \label{eq:dyn}
\end{equation}
with state $x \in \mathbb{R}^n$, control input $u \in \mathbb{R}^m$, nominal dynamics $f: \mathbb{R}^n \times \mathbb{R}_{\geq 0} \rightarrow \mathbb{R}^{n}$, and control input matrix $B: \mathbb{R}^n \times \mathbb{R}_{\geq 0}  \rightarrow \mathbb{R}^{n\times m}$. 
The uncertain dynamics are a linear combination of known regression vectors $\Delta: \mathbb{R}^n \times \mathbb{R}_{\geq 0} \rightarrow \mathbb{R}^{p\times n}$ and unknown parameters $\theta \in \mathbb{R}^p$.
We assume that \cref{eq:dyn} is locally Lipschitz uniformly and that the state $x$ is measured.
The following assumption is made on the unknown parameters $\theta$.

\begin{assumption}
\label{assumption:params}
The unknown parameters $\theta$ belong to a known compact convex set $\Theta \subset \mathbb{R}^p$.
\end{assumption}

An immediate consequence of Assumption~\ref{assumption:params} is that the parameter estimation error $ \ \tilde{\theta} \triangleq \hat{\theta} - \theta \ $ must also belong to a known compact convex set $\tilde{\Theta}$.
Each parameter must then have a finite maximum error where $|\tilde{\theta}_i| \leq \tilde{\vartheta}_i < \infty$ for $i=1,\dots,p$.

\section{Main Results}
\label{sec:main}

\subsection{Overview}
\label{sub:overview}
This section presents the main results of this work.
We first review the definition of the unmatched control Lyapunov function \cite{lopez2021universal} and its role in the proposed approach.
We then derive the first main result followed by several noteworthy extensions, e.g., the scenario where unmatched and matched uncertainties are present in addition to augmenting the adaptation law with model estimation (composite adaptation).
We also present the the so-called leakage modification that guarantees the adaptation gains return to their nominal values asymptotically. 
Conceptually, the proposed approach relies on adjusting the adaptation gain of \emph{individual} parameter estimates online to prevent destabilization induced by parameter adaptation transients.  
This is in stark contrast to \cite{lopez2021universal} where a single gain for \emph{all} parameters was adjusted.
Consequently this work can be considered a more general and natural formulation of \cite{lopez2021universal}.

\subsection{Unmatched Control Lyapunov Functions}
Similar to \cite{lopez2021universal}, the so-called unmatched control Lyapunov function will be used through this letter.
\begin{definition}[cf.~\cite{lopez2021universal}]
\label{def:uclf}
A smooth, positive-definite function $V_\theta : \mathbb{R}^n \times \mathbb{R}^p \times \mathbb{R}_{\geq 0} \rightarrow \mathbb{R}_{\geq 0}$ is an \emph{unmatched control Lyapunov function} (uclf) if it is radially unbounded in $x$ and for each $\theta \in \Theta \subset \mathbb{R}^p$
\begin{equation*}
    \begin{aligned}
        \underset{u \in \mathbb{R}^m}{\mathrm{inf}} \left\{ \frac{\partial V_{\theta}}{\partial t} + \frac{\partial V_\theta}{\partial x}^\top \left[ f - \Delta^\top \theta + B u \right] \right\} \leq  -Q_\theta
    \end{aligned}
\end{equation*}
where $Q_\theta : \mathbb{R}^n \times \mathbb{R}^p \rightarrow \mathbb{R}_{\geq 0}$ is continuously differentiable, {radially unbounded in $x$,} and positive-definite.
\end{definition}

\begin{remark}
    \cref{def:uclf} can also be stated in terms of a contraction metric~\cite{lohmiller1998contraction}, {see \cite[Def. 1]{lopez2021universal}} for more details.
\end{remark}

\Cref{def:uclf} is noteworthy because the existence of an uclf  is equivalent (see \cite[Prop. 1]{lopez2021universal}) to system \cref{eq:dyn} being stabilizable for all $\theta \in \Theta$.
This is the weakest possible requirement one can impose on \cref{eq:dyn} and therefore showcases the generality of the approach. 
\cref{def:uclf} involves constructing a clf for each possible model realization, i.e., a \emph{family} of clf's, {so convergence to the desired equilibrium can be achieved if a suitable adaptation law can be derived.}
{Pragmatically, this can be done analytically, e.g., via backstepping, or numerically via discretization or sum-of-squares programming with the uclf search occurring over the state-parameter space.}
\cref{def:uclf} adopts the certainty equivalence principle philosophy: design a clf (or equivalently a controller) as if the unknown parameters are known and simply replace them with their estimated values online.
This intuitive design approach is not generally possible for nonlinear systems with unmatched uncertainties unless additional robustness properties are imposed or the system be modified -- both representing a departure from certainty equivalence.
The approach developed in this work completely bypasses any additional requirements or system modifications by instead adjusting the adaptation gain online where the so-called adaptation gain update law yields a stable closed-loop system.
In other words, the proposed dynamic adaptation gains method expands the use of the certainty equivalence principle to general nonlinear systems with unmatched uncertainties, simplifying the design of stable adaptive controllers.

\subsection{Dynamic Adaptation Gains}
\label{sec:dag}
The proposed method involves adjusting the adaptation gain---typically denoted as a scalar $\gamma$ or symmetric positive-definite matrix $\Gamma$ is the literature---to prevent the parameter adaptation transients from destabilization the closed-loop system.
It is convenient to express the adaptation gain as a function of a scalar argument that has certain properties.
We will make use of the following definition in deriving an adaptation gain update law that achieves closed-loop stability.

\begin{definition}
\label{def:adag}
    An \emph{admissible dynamic adaptation gain} $\gamma: \mathbb{R} \rightarrow \mathbb{R}_{>0}$ is a function with scalar argument $\rho$ such that $\gamma(0) = \bar{\gamma}$ is the nominal adaptation gain, $\lim_{\rho\rightarrow-\infty} \gamma(\rho) = c > 0$, and $0 < \nabla\gamma(\rho) < \infty$.
\end{definition}
\begin{remark}
    There are several functions at the disposal of the designer when selecting an admissible dynamic adaptation gain. Two examples are $\gamma(\rho) = \bar{\gamma} \, (\frac{0.9}{\rho^2+1} + 0.1)$ or $\gamma(\rho) = \bar{\gamma} \, (0.9\,\mathrm{exp}(\rho/\tau) + 0.1)$.
\end{remark}

With \cref{def:uclf,def:adag}, we are now ready to state the first main theorem of this work.

\begin{theorem}
\label{thm:dyn_gain}
Consider the uncertain system \cref{eq:dyn} with $x_d$ being the desired equilibrium point.
If an uclf $\,V_\theta(x, t)\,$ exists, then $x \rightarrow x_d$ asymptotically with the adaptation law
\begin{equation}
\label{eq:dot_theta_v}
        \dot{\hat{\theta}} = - \, \mathrm{diag}(\gamma_1(\rho_1), \dots, \gamma_p(\rho_p)) \,\Delta(x,t) \frac{\partial V_{\hat{\theta}}}{\partial x} \tag{2a}
\end{equation}
where each $\gamma_i(\cdot)$ is an admissible dynamic adaptation gain whose update law satisfies
\begin{equation}
\label{eq:dot_gamma_V}
    \dot{\gamma}_i(\rho_i) \leq - 2 \frac{\gamma_i(\rho_i)^2}{(\eta_i - \tilde{\theta}^2_i)} \frac{\partial V_{\hat{\theta}}}{\partial \hat{\theta}_i} \, \dot{\hat{\theta}}_i \tag{2b}
\end{equation}
for $\eta_i > \tilde{\vartheta}_i \geq \tilde{\theta}_i$ with $i=1,\dots,p$.
\end{theorem}
\begin{proof}
Consider the Lyapunov-like function
\setcounter{equation}{2}
\begin{equation}
\label{eq:Vc}
    V_c(t) = V_{\hat{\theta}}(x,t) + \frac{1}{2}\sum_{i=1}^p \frac{(\tilde{\theta}_i^2 - \eta_i)}{\gamma_i(\rho_i)},
\end{equation}
where $\tilde{\theta} = \hat{\theta}-\theta$ and $\eta_i > \tilde{\vartheta}_i \geq \tilde{\theta}_i$ {with $\eta_i$ finite}. Differentiating \cref{eq:Vc} yields
\begin{equation*}
\begin{aligned}
    \dot{V}_c(t) = & \,  {\frac{\partial V_{\hat{\theta}}}{\partial t}} + \frac{\partial V_{\hat{\theta}}}{\partial x}^\top[f(x,t) - \Delta(x,t)^\top \theta + {B(x,t)}u] \\
    & + \sum_{i=1}^p\biggl[ \frac{\partial V_{\hat{\theta}}}{\partial \hat{\theta}_i} \, \dot{\hat{\theta}}_i  + \frac{\tilde{\theta}_i}{\gamma_i(\rho_i)}  \dot{\hat{\theta}}_i + \frac{1}{2}\frac{(\eta_i - \tilde{\theta}_i^2)}{\gamma_i(\rho_i)^2} \, \dot{\gamma}_i(\rho_i) \biggr].
\end{aligned}
\end{equation*}
Since $V_{\hat{\theta}}(x,t)$ is an uclf  then
\begin{equation*}
\begin{aligned}
    \dot{V}_c(t) \leq& \, -Q_{\hat{\theta}}(x) + \frac{\partial V_{\hat{\theta}}}{\partial x}^\top \Delta(x,t)^\top \tilde{\theta} \\
    & + \sum_{i=1}^p\biggl[\frac{\partial V_{\hat{\theta}}}{\partial \hat{\theta}_i} \,\dot{\hat{\theta}}_i  + \frac{\tilde{\theta}_i}{\gamma_i(\rho_i)}  \dot{\hat{\theta}}_i + \frac{1}{2}\frac{(\eta_i - \tilde{\theta}_i^2)}{\gamma_i(\rho_i)^2} \, \dot{\gamma}_i(\rho_i) \biggr].
\end{aligned}
\end{equation*}
Noting
\begin{equation*}
    \sum_{i=1}^p \frac{\tilde{\theta}_i}{\gamma_i(\rho_i)}  \dot{\hat{\theta}}_i =  \tilde{\theta}^\top [\mathrm{diag}(\gamma_1(\rho_1),\dots,\gamma_p(\rho_p))]^{-1} \, \dot{\hat{\theta}}
\end{equation*}
then applying \cref{eq:dot_theta_v} yields
\begin{equation*}
\begin{aligned}
    \dot{V}_c(t) & \leq -Q_{\hat{\theta}}(x) + \sum_{i=1}^p\biggl[\frac{\partial V_{\hat{\theta}}}{\partial \hat{\theta}_i} \,\dot{\hat{\theta}}_i + \frac{1}{2} \frac{(\eta_i - \tilde{\theta}_i^2)}{\gamma_i(\rho_i)^2} \, \dot{\gamma}_i(\rho_i) \biggr].
\end{aligned}
\end{equation*}
If $\dot{\gamma}_i(\rho_i)$ is chosen so \cref{eq:dot_gamma_V} is satisfied, then the term in brackets is negative so $\dot{V}_c(t) \leq -Q_{\hat{\theta}}(x) < 0$ if $x \neq x_d$.
Hence, $V_c(t)$ is non-increasing so $V_{\hat{\theta}}(x,t)$ and $\tilde{\theta}$ are bounded since $\gamma_i(\rho_i)$ is lower-bounded and $\eta_i$ is bounded by construction. 
Since $V_{\hat{\theta}}(x,t)$ is radially unbounded then $x$ must also be bounded.
By definition, $Q_{\hat{\theta}}(x)$ is continuously differentiable so $\dot{Q}_{\hat{\theta}}(x)$ is bounded for bounded $x$, $\tilde{\theta}$ and therefore $Q_{\hat{\theta}}(x)$ is uniformly continuous. 
Since $V_c(t)$ is nonincreasing and bounded from below then $\lim_{t\rightarrow \infty} \int_0^t Q_{\hat{\theta}}(x(\tau))\,d\tau  ~{\leq}~ V_c(0) - \lim_{t\rightarrow \infty}V_c(t)$ exists and is finite so by Barbalat's lemma $\lim_{t\rightarrow \infty}  Q_{\hat{\theta}}(x(t)) \rightarrow 0$.
Recalling $Q_{\hat{\theta}}(x) = 0 \iff x = x_d$ then we can conclude $x(t) \rightarrow x_d$ as $t \rightarrow \infty$ as desired. \qedhere
\end{proof}
\begin{remark}
    \label{remark:proj}
    Since the unknown parameters belong to a compact convex set $\Theta$ then one can employ the projection operator $\mathrm{Proj}_\Theta (\cdot)$ to ensure $\hat{\theta} \in \Theta$ without affecting stability \cite{slotine1991applied,ioannou2012robust}.
\end{remark}

\begin{remark}
    \label{remark:rho}
    \Cref{eq:dot_gamma_V} is expressed as an update to $\dot{\gamma}_i(\rho_i)$ to make the dynamic adaptation gain interpretation more obvious.
    For implementation, one should rewrite \cref{eq:dot_gamma_V} as a bound on $\dot{\rho}_i$ via the chain rule.
    Note $\nabla\gamma_i(\rho_i)$ in invertible by \cref{def:adag}.
\end{remark}

Before deriving various implementable forms of the dynamic adaptation gain update law, it is instructive to analyze its general behavior.
For the case where a parameter's adaptation transients is \emph{destabilizing}, i.e., $\tfrac{\partial V_{\hat{\theta}}}{\partial \hat{\theta}_i} \tfrac{d}{dt}\hat{\theta}_i > 0$ for some $i$, we see from \cref{eq:dot_gamma_V} that $\dot{\gamma}_i(\rho_i) < 0$ so the adaptation gain decreases to prevent destabilization. 
In other words, the adaptation rate is slowed for parameters whose adaptation transients could cause instability.
Note that there is no concern of the adaptation gain $\gamma_i(\rho_i)$ becoming negative since $\gamma_i(\rho_i)$ must be an admissible dynamic adaptation gain so, by definition, $\gamma_i(\rho_i)$ is lower-bounded by a positive constant. 
With that said, $\rho_i$ could become a large negative number which does not affect the proof of \cref{thm:dyn_gain} but can have practical implications; this will be discussed in more detail below.
Now considering the case where a parameter's adaptation transients is \emph{stabilizing}, i.e., $\tfrac{\partial V_{\hat{\theta}}}{\partial \hat{\theta}_i} \tfrac{d}{dt}\hat{\theta}_i \leq 0$ for some $i$, we see from \cref{eq:dot_gamma_V} that $\dot{\gamma}_i(\rho_i)$ is simply upper-bounded by a positive quantity. 
Consequently, one can set $\dot{\gamma}_i(\rho_i) = 0$ and still obtain a stable closed-loop system since the transients term is negative thereby leading to the desired stability inequality in the proof of \cref{thm:dyn_gain}. 
Another choice would be to let $\dot{\gamma}_i(\rho_i) > 0$ (without exceeding the inequality \cref{eq:dot_gamma_V}) until ${\gamma}_i(\rho_i) = \bar{\gamma}_i$, i.e., the nominal adaptation gain value, at which point one sets $\dot{\gamma}_i(\rho_i) = 0$.
This leads us to one \emph{implementable} form of the adaptation gain update law.

\begin{corollary}
\label{cor:update_law_1}
    The adaptation gain update law given by
    \begin{equation*}
        \dot{\gamma}_i(\rho_i) = \begin{cases} 
            - \frac{2\bar{\gamma}_i^2}{(\eta_i - \tilde{\vartheta}^2_i)} \frac{\partial V_{\hat{\theta}}}{\partial \hat{\theta}_i} \, \dot{\hat{\theta}}_i ~~ \mathrm{when} ~ \frac{\partial V_{\hat{\theta}}}{\partial \hat{\theta}_i} \, \dot{\hat{\theta}}_i > 0 \\[6pt]
            - \frac{2c_i^2}{\eta_i} \frac{\partial V_{\hat{\theta}}}{\partial \hat{\theta}_i} \, \dot{\hat{\theta}}_i ~~ \mathrm{when}~\frac{\partial V_{\hat{\theta}}}{\partial \hat{\theta}_i} \, \dot{\hat{\theta}}_i \leq 0~\mathrm{and}~\gamma_i(\rho_i) < \bar{\gamma}_i \\[6pt]
            \hphantom{-} 0  ~~ \mathrm{when}~\frac{\partial V_{\hat{\theta}}}{\partial \hat{\theta}_i} \, \dot{\hat{\theta}}_i \leq 0 ~ \mathrm{and} ~\gamma_i(\rho_i) = \bar{\gamma}_i,
        \end{cases}
    \end{equation*}
    for $i=1,\dots,p$ satisfies condition \cref{eq:dot_gamma_V} in \cref{thm:dyn_gain} so $x \rightarrow x_d$ asymptotically as desired.
\end{corollary}

The component-wise nature of the adaptation gain update law presented in \cref{cor:update_law_1} systematically checks the adaptation transients of each parameter and adjusts the adaptation gain accordingly to preserve stability.
Algorithmically, the update law cycles through each parameter and updates the adaptation gain on the parameters whose transients are problematic. 
Treating parameters individually rather than as a whole (as in \cite{lopez2021universal}, see the \nameref{sec:appendix}) allows for \emph{targeted} adjustments to problematic parameters which yields less myopic behavior of the controller and better closed-loop performance.
Note this behavior is quite different than that in \cite{lei2006universal} where the gain was only increased to achieve stability.

\begin{remark}
\label{remark:mult}
    \cref{cor:update_law_1} for $\frac{\partial V_{\hat{\theta}}}{\partial \hat{\theta}_i} \, \tfrac{d}{dt}{\hat{\theta}}_i \leq 0$ can be rewritten as
    \begin{equation*}
    \begin{aligned}
        \dot{\gamma}_i(\rho_i) \ &= \ - \frac{2 c^2_i}{\eta_i} \ \frac{\partial}{\partial \hat{\theta}_i} \bigl[\log(h(x)(V_{\hat{\theta}}(x,t) + c))\bigr] \, \dot{\hat{\theta}}_i
    \end{aligned}
    \end{equation*}
    where $h: \mathbb{R}^n \rightarrow \mathbb{R}_{>0}$ and $c \in \mathbb{R}_{>0}$. 
    This form shows that the adaptation gain update is unaffected if an uclf  is scaled by a uniformly positive function $h(x)$.
    This could have implications for safety-critical adaptive control if $h(x)$ were akin to a barrier function.
    If $h(x)$ were also to be upper-bounded then a similar relationship can be obtained for the $\frac{\partial V_{\hat{\theta}}}{\partial \hat{\theta}_i} \, \tfrac{d}{dt}{\hat{\theta}}_i > 0$.
    The gradient term above is similar to the score function $\nabla \log p(x) $ in diffusion-based generative modeling~\cite{song}, where $p(x)$ a probability density known only within a scaling factor (the partition function) which disappears in $\nabla \log p(x) $.
\end{remark}

\subsection{{Leakage Modification: Bounding $\rho$}}
\label{sub:leak}
{The proof of \cref{thm:dyn_gain} requires the individual adaptation gains be lower-bounded.
This is accomplished through appropriate selection of each $\gamma_i(\cdot)$ to meet the conditions of \cref{def:adag}.
It is desirable to ensure the scalar argument $\rho_i$ remains bounded through means beyond simple parameter tuning to prevent numerical instability cause by $|\rho_i|$ becoming too large.
Having the adaptation gains return to their nominal value automatically after the adaptation transients has subsided is also ideal.
We propose to replace the pure integrator $\rho_i$ dynamics with nonlinear first-order dynamics, i.e., a leakage modification, to achieve these desirable properties.
The modified adaptation gain update law takes the form
\begin{equation}
\label{eq:leak}
    \dot{\rho}_i = 2 \frac{\gamma_i(\rho_i)^2}{\nabla \gamma_i(\rho_i)} \big[  - \lambda_i  \, \rho_i + K_i \, w_i(x) \big] \, ,
\end{equation}
where $w_i(x) = \Bigl[\frac{\partial V_{\hat{\theta}}}{\partial \hat{\theta}}^\top \Delta(x,t) \frac{\partial V_{\hat{\theta}}}{\partial x}\Bigr]_i$ with $[\cdot]_i$ the $i$-th component, $\, K_i = \nicefrac{\bar{\gamma}_i}{(\eta_i - \tilde{\vartheta}^2_i)}\, $ if $\, w_i(x) < 0\, $ and $\, K_i = 0$ otherwise, and $\lambda_i \in \mathbb{R}_{>0}$.  
Note $w_i(x) < 0$ corresponds to the transients being \emph{destabilizing}, so the input of \cref{eq:leak} is zero if the transients is stabilizing or has subsided. 
The following lemma establishes important properties of \cref{eq:leak}.
\begin{lemma}
\label{lemma:filter}
    The output of the dynamic adaptation gain update law \cref{eq:leak} remains bounded if the exogenous signal $w_i(x)$ is bounded. Furthermore, the output tends to zero if $w_i(x) \rightarrow 0$.
\end{lemma}
The proof can be found in the \nameref{sec:appendix}. We now show \cref{eq:leak} when combined with \cref{eq:dot_theta_v} yields a stable closed-loop system.
\begin{theorem}
\label{thm:leak}
    Consider the uncertain system \cref{eq:dyn} with $x_d$ being the desired equilibrium point.
    If an uclf $\,V_\theta(x, t)\,$ exists, then $x \rightarrow x_d$ asymptotically with the adaptation law \cref{eq:dot_theta_v} and dynamic adaptation gain update law \cref{eq:leak}. 
    Furthermore, $\rho_i \rightarrow 0$ and in turn $\gamma_i(\rho_i) \rightarrow \bar{\gamma}_i$ asymptotically for each $i=1,\dots,p$.  
\end{theorem}
\begin{proof}
    Consider the Lyapunov-like function
    \begin{equation*}
        V_c(t) = V_{\hat{\theta}}(x,t) + \sum_{i=1}^p \eta_i \, \lambda_i \int\limits_t^T |\rho_i(\tau)|\, d\tau + \frac{1}{2}\sum_{i=1}^p \frac{(\tilde{\theta}_i^2 - \eta_i)}{\gamma_i(\rho_i)},
    \end{equation*}
    where $\eta_i$ and $\lambda_i$ are defined as before and $T$ is sufficiently large \cite{kalman1960control,luenberger1979introduction} but finite to ensure the integral is finite. 
    Differentiating and applying \cref{def:uclf,eq:dot_theta_v,eq:leak},
    \begin{equation*}
        \begin{aligned}
            \dot{V}_c(t) & \leq  - Q_{\hat{\theta}}(x) - \sum_{i=1}^p \eta_i \, \lambda_i \, |\rho_i|  \\
            & \hphantom{\leq} + \sum_{i=1}^p \left[ \frac{\partial V_{\hat{\theta}}}{\partial \hat{\theta}_i} \dot{\hat{\theta}}_i + (\eta_i - \tilde{\theta}_i^2) [-\lambda_i\, \rho_i + K_i \, w_i(x) ] \right] \\
            &\leq - Q_{\hat{\theta}}(x) - \sum_{i=1}^p \left[ \eta_i \, \lambda_i \, |\rho_i| + (\eta_i - \tilde{\theta}_i^2) \, \lambda_i\, \rho_i  \right],
        \end{aligned}
    \end{equation*}
    where the second inequality holds by choice of $K_i$.
    Also by choice of $K_i$, \cref{eq:leak} is only driven by an input that is either negative or zero, so $\rho_i(t) \leq 0$ for all $t\geq 0$ since $\rho_i(0) = 0$.
    Hence, the term in brackets must be negative which yields $\dot{V}_c(t) \leq - Q_{\hat{\theta}}(x)$ so $V_c(t)$ is nonincreasing and in turn $x$ and $\tilde{\theta}$ are bounded. 
    Using the same arguments as in \cref{thm:dyn_gain} and noting $V_c(t)$ is still lower-bounded, $x(t) \rightarrow x_d$ as $t \rightarrow \infty$ via Barbalat's lemma.
    Since $\tfrac{\partial V_{\hat{\theta}}}{\partial x} \rightarrow 0 \iff x \rightarrow x_d$ by construction, then each $w_i(x) \rightarrow 0$ as $x \rightarrow x_d$.
    Hence, by \cref{lemma:filter}, $\rho_i \rightarrow 0$ and in turn $\gamma_i(\rho_i) \rightarrow \bar{\gamma}_i$ for $i=1,\dots,p$
\end{proof} }

\subsection{Matched and Unmatched Uncertainties}
A particularly useful property of the dynamic adaptation gains method is the ability to treat matched and unmatched uncertainties separately. 
{This is in stark contrast to the approach taken in \cite{lopez2021universal} where all adaptation gains were adjusted to cancel adaptation transients.}
The separability inherent to the current method is indicative of the more natural formalism of treating adaptation transients on an individual basis rather than as a whole. 
The following theorem solidifies this point.
\begin{theorem}
\label{thm:matched}
    Assume system \cref{eq:dyn} can be rewritten as
    \begin{equation}
    \label{eq:dyn2}
        \dot{x} = f(x,t) - \Delta(x,t)^\top \theta + B(x,t) [u - \Psi(x,t)^\top \phi]
    \end{equation}
    where $\phi \in \Phi \subset \mathbb{R}^q$ are the matched parameters with known regression vectors $\Psi:\mathbb{R}^n \times \mathbb{R}_{\geq 0} \rightarrow \mathbb{R}^{q \times n}$.
    Let $x_d$ denote the desired equilibrium point.
    If an uclf $\,V_\theta(x, t)\,$ exists with corresponding control law $u_{\theta}$, then $x \rightarrow x_d$ asymptotically with the controller $\kappa = u_{\hat{\theta}} + \Psi(x,t)^\top \hat{\phi}$ and adaptation laws
\begin{equation}
\label{eq:matched-adapt}
    \begin{aligned}
        \dot{\hat{\phi}} &= - \, \Gamma \, B(x,t)^\top \Psi(x,t) \, \frac{\partial V_{\hat{\theta}}}{\partial x} \\
        \dot{\hat{\theta}} &= -  \, \mathrm{diag}(\gamma_1(\rho_1), \dots, \gamma_p(\rho_p)) \,\Delta(x,t) \, \frac{\partial V_{\hat{\theta}}}{\partial x}
    \end{aligned}
\end{equation}
where $\Gamma$ is a constant symmetric positive-definite matrix and each $\gamma_i(\cdot)$ is an admissible dynamic adaptation gain.
\end{theorem}

\begin{proof}
    Consider the new Lyapunov-like function 
    \begin{equation*}
        V_c(t) = V_{\hat{\theta}}(x,t) + \frac{1}{2}\tilde{\phi}^\top \Gamma^{-1} \tilde{\phi} + \frac{1}{2}\sum_{i=1}^p \frac{(\tilde{\theta}_i^2-\eta_i)}{\gamma_i(\rho_i)},
    \end{equation*}
    where $\tilde{\phi} \triangleq \hat{\phi}-\phi$ and $\tilde{\theta},~\eta_i$ are defined as before. Differentiating $V_c(t)$ along \cref{eq:dyn2} and substituting $\kappa$ yields
    \begin{equation*}
        \begin{aligned}
            \dot{V}_c(t) \leq & \, -Q_{\hat{\theta}}(x) + \frac{\partial V_{\hat{\theta}}}{\partial x}^\top \Delta(x,t)^\top \tilde{\theta} \\
            & + \frac{\partial V_{\hat{\theta}}}{\partial x}^\top \Psi(x,t)^\top {B(x,t)} \tilde{\phi} + \tilde{\phi}^\top \Gamma^{-1} \, \dot{\phi}_i \\
            & + \sum_{i=1}^p\biggl[\frac{\partial V_{\hat{\theta}}}{\partial \hat{\theta}_i}\,\dot{\hat{\theta}}_i  + \frac{\tilde{\theta}_i}{\gamma_i(\rho_i)} \dot{\hat{\theta}}_i + \frac{1}{2}\frac{(\eta_i - \tilde{\theta}_i^2)}{\gamma_i(\rho_i)^2} \, \dot{\gamma}_i(\rho_i) \biggr].
        \end{aligned}
    \end{equation*}
    Substituting in \cref{eq:matched-adapt,eq:dot_gamma_V} yields $\dot{V}_c(t) \leq - Q_{\hat{\theta}}(x)$ so $V_c(t)$ is nonincreasing. 
    Similar to \cref{thm:dyn_gain}, we can conclude that $\lim_{t\rightarrow\infty} Q_{\hat{\theta}}(x(t)) \rightarrow 0 \implies x(t) \rightarrow x_d$ as $t\rightarrow\infty$.
\end{proof}

\cref{thm:matched} can be immediately extended to use other well-known direct matched adaptation laws, e.g., those based on sliding variables \cite{slotine1991applied}, nonlinear damping \cite{krstic1995nonlinear}, or any other suitable adaptation law.
This further highlights the versatility of individual dynamic adaptation gains compared to \cite{lopez2021universal}.

\subsection{Composite Adaptation}
\label{sub:comp}
Parameter adaptation transients can be markedly improved by combining direct and indirect adaptation schemes to form a composite adaptation law \cite{slotine1989composite}.
The following proposition shows composite adaptation with dynamic adaptation gains also yields a stable closed-loop system.
\begin{proposition}
\label{prop:composite}
    Assume a signal $\varepsilon_{\hat{\theta}} = W(x,t)^\top\tilde{\theta}$ for some matrix $W(x,t)$ is available via measurement or computation. 
    If an uclf $\,V_\theta(x, t)\,$ exists, then $x \rightarrow x_d$ asymptotically with the composite adaptation law
\begin{equation*}
    \dot{\hat{\theta}} = - \mathrm{diag}(\gamma_1(\rho_1), \dots, \gamma_p(\rho_p)) \Bigl(\Delta(x,t) \frac{\partial V_{\hat{\theta}}}{\partial x} + \beta \, W(x,t) \, \varepsilon_{\hat{\theta}} \Bigr)
\end{equation*}
where $\beta \in \mathbb{R}_{>0}$ and each $\gamma_i$ is an admissible dynamic adaptation gain whose update law satisfies 
\cref{eq:dot_gamma_V}.
\end{proposition}
\begin{proof}
    Using the Lyapunov-like function from \cref{thm:dyn_gain} and applying the composite adaptation law with \cref{eq:dot_gamma_V} yields $\dot{V}_c(t) \leq -Q_{\hat{\theta}}(x) - \beta \, \tilde{\theta}^\top W(x,t) \, \varepsilon_{\hat{\theta}} = -Q_{\hat{\theta}}(x) - \beta \, \tilde{\theta}^\top W(x,t) W(x,t)^\top \tilde{\theta} \implies \dot{V}_c(t) \leq -Q_{\hat{\theta}}(x)$ so $V_c(t)$ is nonincreasing.
    Similar to \cref{thm:dyn_gain}, we can conclude that $\lim_{t\rightarrow\infty} Q_{\hat{\theta}}(x(t)) \rightarrow 0 \implies x(t) \rightarrow x_d$ as $t\rightarrow\infty$.
\end{proof}

\section{Simulation Experiments}

The developed method was tested on the system
\begin{equation}
\label{eq:ex_dyn_2}
\left[ \begin{array}{c} \dot{x}_1 \\ \dot{x}_2 \\ \dot{x}_3 \end{array} \right] = \left[ \begin{array}{c}  x_3 - \theta_1 x_1 \\ - x_2 -\theta_2 x^2_1  \\ \mathrm{tanh}(x_2) - \theta_3 x_3 - \theta_4 x_1^2 \end{array} \right] + \left[ \begin{array}{c} 0 \\ 0 \\ 1 \end{array} \right] u,
\end{equation}
with state $x = [x_1,~x_2,~x_3]^\top$, unknown parameters $\theta = [\theta_1,~\theta_2,~\theta_3,~\theta_4]^\top$ where $\theta_{1},\,\theta_2$ are unmatched, and $x_d$ being the origin.
The system is not feedback linearizable and is not in strict feedback form.
The true model parameters are $\theta^* = [-1.8,\,-2.4,\,-0.75,\,-2.25]^\top$ with the set of allowable variations ${\theta} \in [-2.1,\,1.5] \times [-3,\,1.5] \times [-1.8,\,2.25] \times [-5.25,\,1.5]$; the projection operator from \cite{slotine1991applied} was used to bound each parameter.
The dynamic adaptation gains took the form $\gamma_i(\rho_i) = 0.9\, e^{\rho_i} + 0.1$ and $\eta_i = 10+\vartheta_i^2$ for $i=1,2$.
Similar to \cite{lopez2021universal}, the Riemannian energy of a geodesic connecting $x$ and $x_d$ was chosen to be the uclf.

\Cref{fig:E} shows the uclf---an indicator of tracking performance---with and without adaptation.
The proposed method (\cref{cor:update_law_1}) successfully stabilizes the origin as predicted by \cref{thm:dyn_gain}.
This is in stark contrast to the no adaptation case where only bounded error can be acheived.
The uclf with leakage modification (omitted for clarity) exhibits nearly identical behavior as the uclf with the update law from \cref{cor:update_law_1} which confirms the result stated in \cref{thm:leak}. 
\Cref{fig:gains} shows the individual adaptation gains with and without the leakage modification. 
The nominal case (blue) sees a reduction in the adaptations gains by 47\% and 9\%, respectively, indicating the adaptation transients of $\theta_1$ would have a large destabilizing effect if not compensated for.
The leakage modification (red, $\lambda = 1$) ensures the adaptation gains return to their nominal values as predicted by \cref{thm:leak}.

\begin{figure}[t!]
\centering 
  \subfloat[Unmatched clf.]{\includegraphics[trim=100 0 180 0, clip, width=.48\columnwidth]{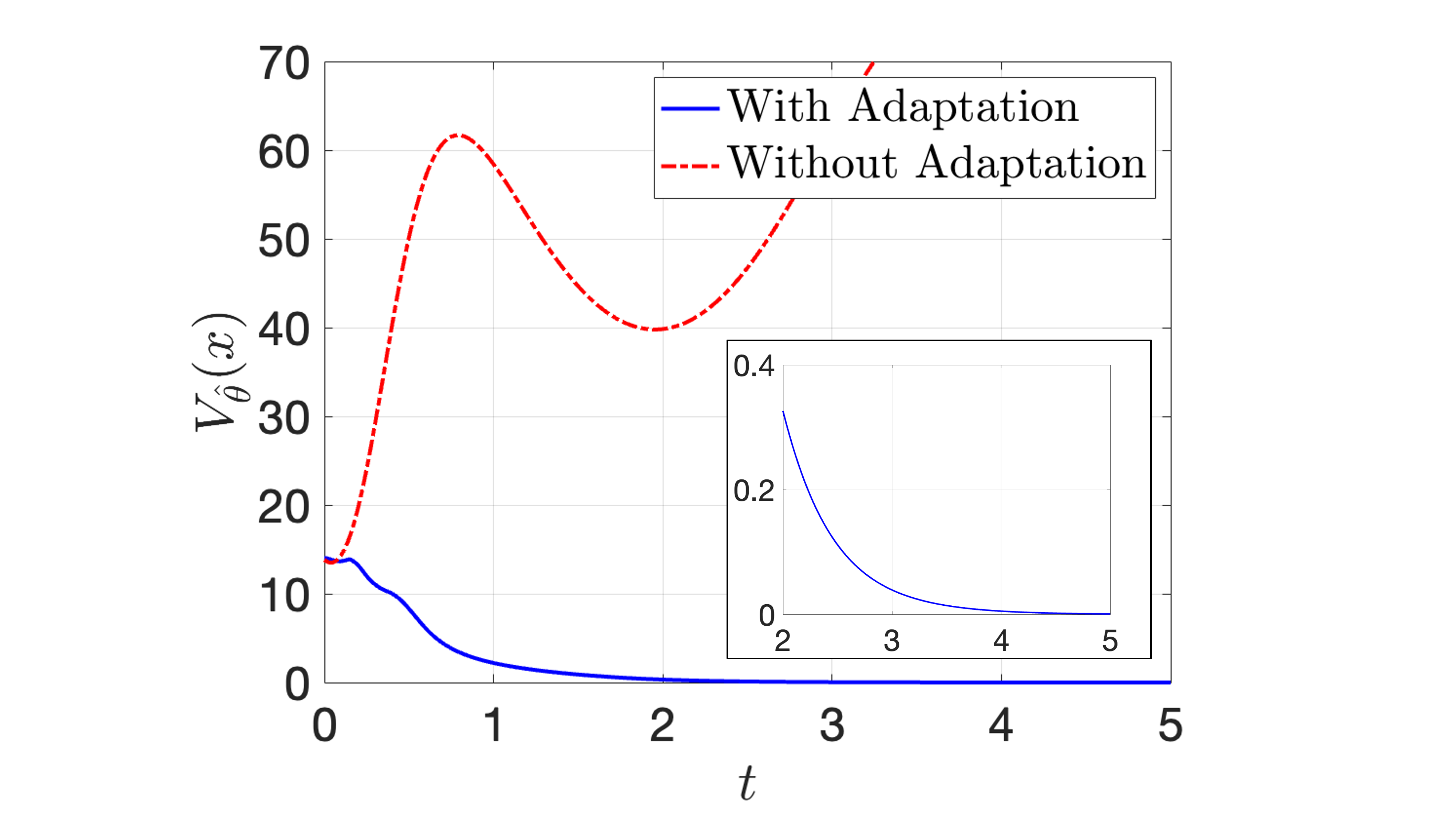}\label{fig:E}}
  \hspace{0.3em}
  \subfloat[Adaptation gains.]{\includegraphics[trim=100 0 180 0, clip, width=.48\columnwidth]{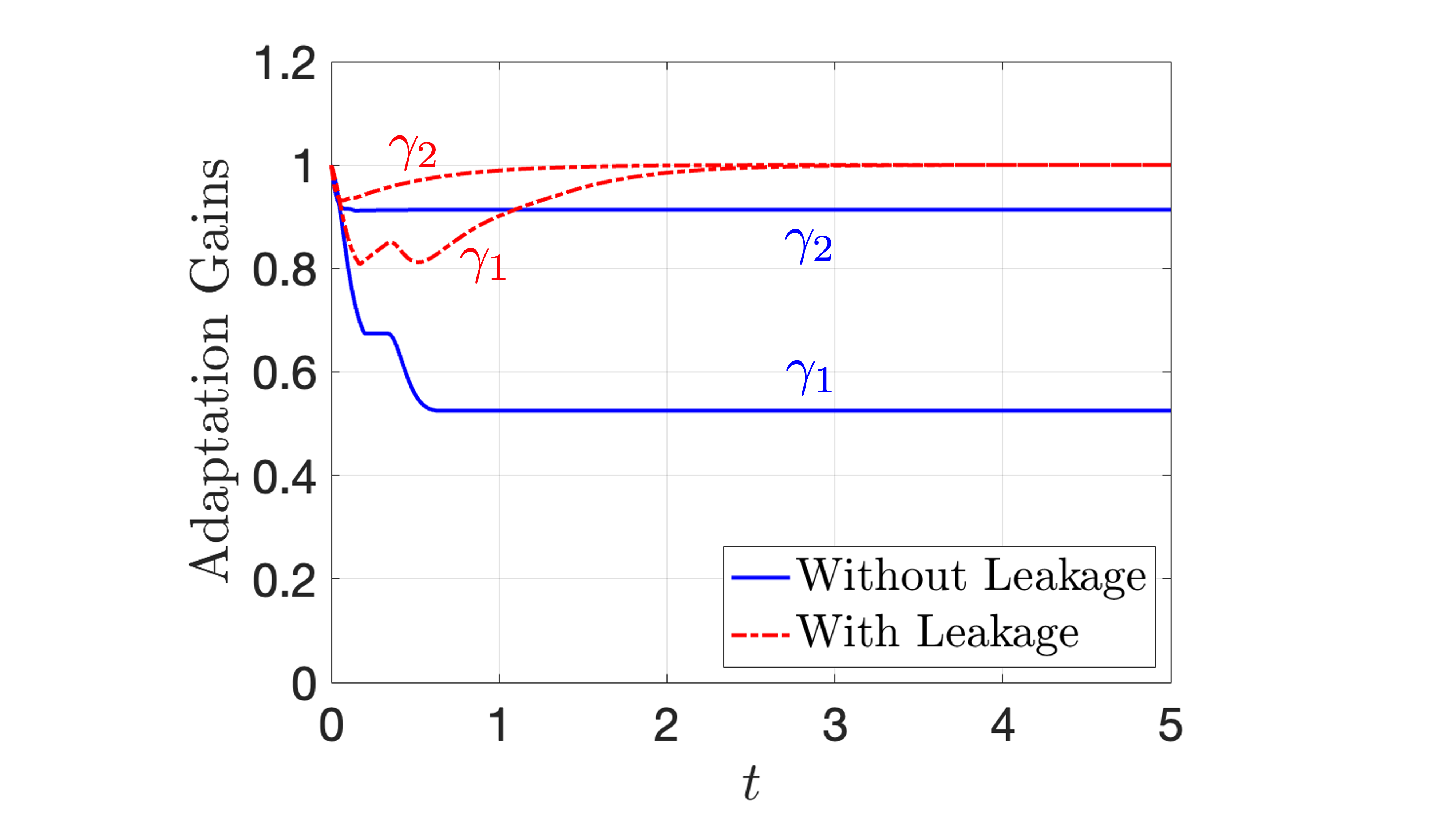}\label{fig:gains}}
  \caption{Performance and behavior of the dynamic adaptation gains method. (a): The uclf with and without adaptation shows the method is able to stabilize the origin while the system is unstable with no adaptation. (b): The adaptation gains are reduced by as much as 47\% to ensure stability. When the leakage modification is added the gains return to their nominal values as desired.} 
    \label{fig:results}
    \vskip -0.2in
\end{figure}

\section{Conclusion}
We presented a new direct adaptation law that adjusts individual adaptation gains online to achieve stable closed-loop control and learning for nonlinear systems with unmatched uncertainties.
A bound on the rate of change of individual adaptation gains that prevents destabilization by the adaptation transients was derived.
Noteworthy extensions and modifications were also discussed.
The results presented here have important implications in adaptive safety \cite{lopez2023unmatched} and adaptive optimal control \cite{lopez2021adaptive}.
Future work will explore these implications in addition to deeper investigations into fundmanetal properties of the approach.

\section*{Appendix}
\label{sec:appendix}
\begin{proof}[{Proof of \cref{lemma:filter}}]
    {
    Consider the time interval $t \in [t_0,\,t_1]$ where $\bar{w} \triangleq {\sup}_{t\in[t_0,\,t_1]}(|Kw(x(t))|)$ is the supremum of the input to \cref{eq:leak} over the time interval of interest \cite{krstic1995nonlinear}.
    Note the $i$ subscript is dropped for clarity.
    Consider the comparison system $\dot{z} = - \lambda a(z) z + a(z) \bar{w}$ where $|\rho| \leq z$ uniformly and $a(z) \triangleq 2 \frac{\gamma(z)^2}{\nabla \gamma(z)}$ is uniformly strictly positive. 
    The comparison system has the virtual dynamics $\dot{y} = - \lambda a(z) y + a(z) \bar{w}$ which is contracting in $y$ so any two particular solutions converge exponentially to each other \cite{wang2005partial}.
    Letting $\alpha(t) \triangleq \lambda \, \int_{t_0}^t a(\rho(\tau))\, d \tau > 0$, then $y(t) = y(t_0) e^{-\alpha(t)} + \tfrac{1}{\lambda}(1-e^{-\alpha(t)})\,\bar{w}$ is the solution to the virtual dynamics and is bounded.
    Because $z$ is a particular solution of the virtual system and $|\rho| \leq z$, $\rho$ is bounded.
    Also, as $w(x) \rightarrow 0$ then so does $\bar{w}$ over some time interval, hence $y \rightarrow 0$ and in turn $\rho \rightarrow 0$ as desired.
    }
\end{proof}
\begin{theorem}[cf.~\cite{lopez2021universal}]
    Consider the uncertain system \cref{eq:dyn} with $x_d$ being the desired equilibrium.
    If an uclf \, $V_\theta(x, t)$ \, exists, then, for any strictly-increasing and uniformly-positive scalar function $\upsilon(\rho)$, $x \rightarrow x_d$ asymptotically with the adaptation law
    \begin{equation*}
        \begin{aligned}
            \dot{\hat{\theta}} & = - \upsilon(\rho) \,  \Gamma \,  \Delta(x,t) \, \frac{\partial V_{\hat{\theta}}}{\partial x}, \\
            \dot{\rho} & = - \frac{\upsilon(\rho)}{\nabla \upsilon(\rho)} \frac{1}{V_{\hat{\theta}}(x,t) + c} \frac{\partial V_{\hat{\theta}}}{\partial \hat{\theta}}^\top \, \dot{\hat{\theta}},
        \end{aligned}
    \end{equation*}
    where $\Gamma$ is a symmetric positive-definite matrix and $c \in \mathbb{R}_{>0}$.
\end{theorem}

\vspace{0.1in}
\noindent \textbf{Acknowledgements} \ \  We thank Miroslav Krstic for stimulating discussions.
\vspace{-0.2in}


\bibliographystyle{ieeetr}
\bibliography{ref}

\end{document}